\documentclass[10pt,journal,cspaper,compsoc]{IEEEtran}
\ifCLASSOPTIONcompsoc
\else
\fi
\ifCLASSINFOpdf
\else
\fi

\usepackage{color,soul}
\usepackage{graphicx}
\usepackage{epstopdf}
\usepackage{wrapfig}

\newtheorem{theoa}{Theorem}

\begin{document}
%
\title{On Linear Spaces of Polyhedral Meshes}
%
%
%
%

\author{Roi Poranne \quad  Renjie Chen \quad Craig Gotsman
\\
Technion -- Israel Institute of Technology
}

%
%

\markboth{}%
{}
%


\IEEEcompsoctitleabstractindextext{%
\begin{abstract}
Polyhedral meshes (PM) -  meshes having planar faces - have enjoyed a rise in popularity in recent years due to their importance in architectural and industrial design. However, they are also notoriously difficult to generate and manipulate. Previous methods start with a smooth surface and then apply elaborate meshing schemes to create polyhedral meshes approximating the surface. In this paper, we describe a reverse approach: given the topology of a mesh, we explore the space of possible planar meshes with that topology.

Our approach is based on a complete characterization of the maximal linear spaces of polyhedral meshes contained in the curved manifold of polyhedral meshes with a given topology. We show that these linear spaces can be described as nullspaces of differential operators, much like harmonic functions are nullspaces of the Laplacian operator. An analysis of this operator provides tools for global and local design of a polyhedral mesh, which fully expose the geometric possibilities and limitations of the given topology.
\end{abstract}

}

\maketitle



\IEEEdisplaynotcompsoctitleabstractindextext

%
\IEEEpeerreviewmaketitle

\begin{figure*}[th]
\centering
\includegraphics[width=0.95\linewidth]{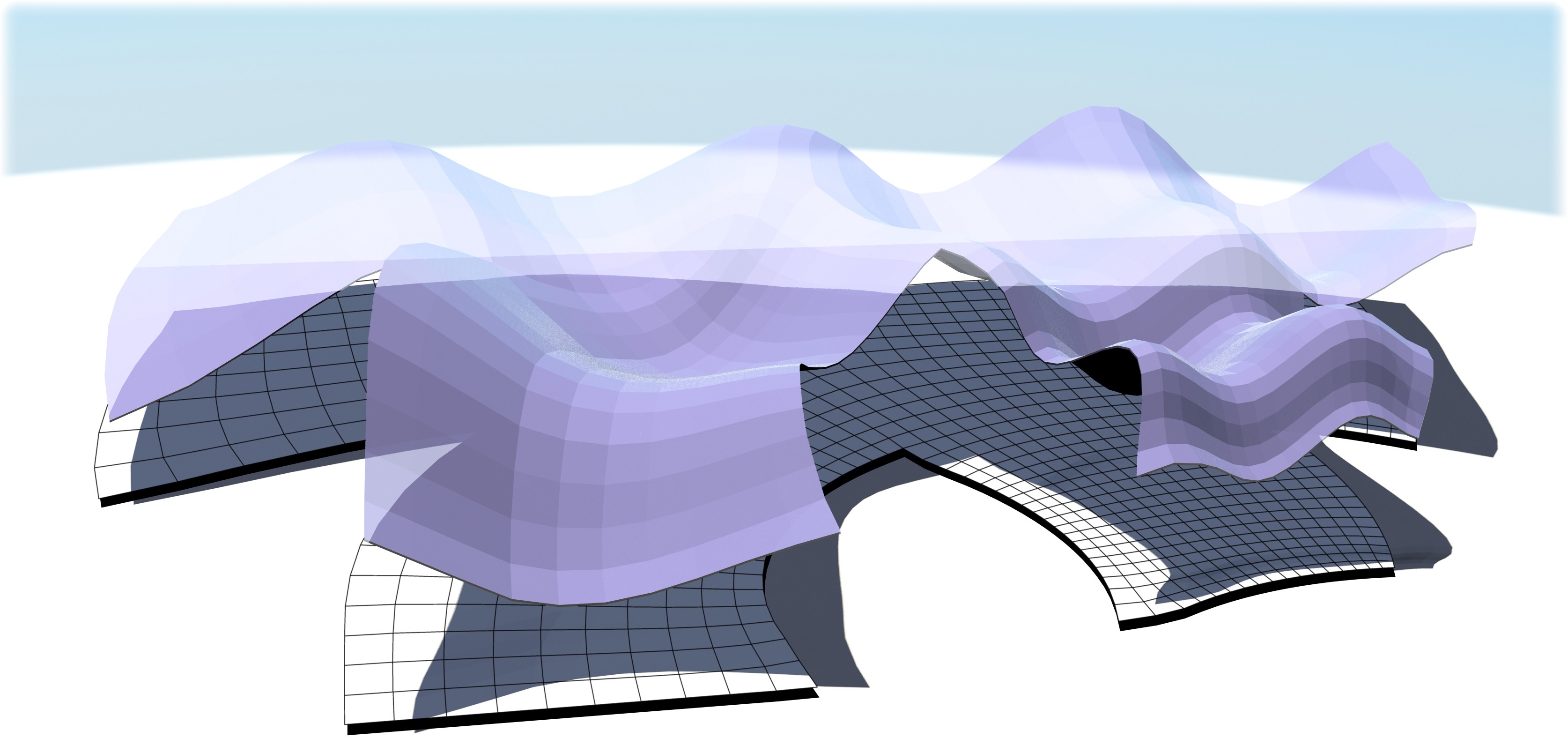}
\caption{A polyhedral mesh constructed from a planar graph using maximal linear subspaces.}
\label{fig:teaser}
\end{figure*}

\section{Introduction}

Polyhedral Meshes (PM's) have gained popularity in recent years due to several new methods that render their construction relatively easy. Typically, a designer creates a traditional free-form surface and then applies a meshing scheme that generates an approximating mesh consisting of only planar faces. Of course, the focus of these schemes, e.g. \cite{Liu1,ZSW11}, is to generate good approximations, and this is done using very specific (regular) types of mesh topologies. It may well be that these are the \emph{only} topologies that can approximate general smooth surfaces well. However, the topology of the mesh itself has its own artistic value: a triangular meshing of a surface will not have the same "look" as a quad or hex meshing. Yet, as mentioned, the cases where a smooth surface can be faithfully meshed into a PM are limited. Hence, we propose a different strategy: instead of constructing the final PM based on a design of a surface, we explore the space of possible PM's with a given topology. Such a PM is called a \emph{realization} of the topology. 

Our approach is based on the observation that the complicated manifold of PM's with a given topology can be decomposed into overlapping, linear spaces, each of which is \emph{maximal} - adding a base PM to the space will introduce non-PMs to the space. The advantage of linear spaces lies in the simplicity of exploring them: PM's in such a space may be designed by forming linear combinations of a spanning set of basic PM's. The disadvantage is that the dimensionality of these spaces is much smaller than that of the complete manifold of PM's. Thus, proving that they are indeed maximal is crucial. By switching between spaces, it is possible to reach any PM in the manifold. We will refer to the PM's of a spanning set simply as \emph{shapes}. 

The use of linear spaces can be incorporated into well-known mesh deformation methods, such as as-rigid/similar-as-possible \cite{Igarashi:2005:ASM:1186822.1073323}. In addition, we propose three types of shapes aimed at different levels of design, exposing the possibilities and limitations for deforming a given PM; the reason for their names will subsequently become clear. \emph{Eigenshapes} are globally smooth shapes at different \emph{frequencies} akin to the eigenvectors of the Laplacian. \emph{Sparse shapes} are based on the smallest groups of vertices that can move together without impairing the planarity of the faces of the PM. Finally, \emph{fundamental shapes} allow a single vertex to be moved with minimal change to other vertices while preserving planarity.

\subsection{Related  work}
\textbf{Meshing and Planarization}. The creation of polyhedral meshes is an active field of research. The most common problem is to mesh, or remesh, a free-form into a PM. The approach used by Cohen-Steiner et al. \cite{Cohen-Steiner} is to try to fit a limited number of planes to the surface and then intersect them. The surface is first partitioned into a user-defined number of \emph{almost} flat regions, for each of which a plane is fitted. These planes, called \emph{shapes proxies}, will generally not have well-defined intersection points. Thus, the faces they produce are only \emph{close} to being planar. Cutler and Whiting \cite{Cutler} added an iterative optimization process to the algorithm that guarantees that the resulting faces are planar. 

In both of these systems, the user can control the number of faces and their density in the result, but cannot dictate the mesh topology (its edge structure), which can essentially be arbitrary. While this is not necessarily a drawback, in some cases a regular mesh is desirable. Liu et al. \cite{Liu1} and Wang et al. \cite{wang2008hexagonal} showed how a surface may be meshed into a planar quad-dominant (PQ) mesh and a planar hexagonal (P-Hex) mesh, respectively. The two algorithms are quite similar: an almost polyhedral mesh is first generated from the surface, based on differential geometric entities (PQ meshes are based on conjugate networks and P-Hex meshes on the Dupin indicatrix. \cite{ZSW11} and \cite{LXW11} elaborated on how to design better conjugate networks.) A subsequent step involves the planarization of the result: a non-linear optimization, where the vertices of the mesh are repositioned to make the faces planar. This latter step seems to dominate the runtime, and does not scale well with mesh size. Alexa and Wardetzky \cite{AW2011} demonstrated the construction of a Laplacian operator on non-triangular meshes. As a side effect of their construction, they were able to devise a related operator that measures the planarity of faces. With this new operator, they obtained a \emph{planarizing flow}, that is, a geometric flow that flattens faces. In \cite{planarization}, a local/global based alternating algorithm was used to solve the planarization problem very efficiently. The improved performance enables interactive deformation of PM's.

\textbf{Mesh deformation}. The problem of editing and deforming mesh geometry is one of the most studied topics in geometry processing. Most mesh deformation methods are intended to work exclusively with triangle meshes. See \cite{PMP:2010} for a thorough introduction. These methods may be classified into two types, based on the type of user interaction employed. In the first type, the user directly modifies the surface using one of a number of common design metaphors. The most relevant to us are the handle-based methods (e.g. \cite{Ben-Chen:2009:VHM:1531326.1531340, Igarashi:2005:ASM:1186822.1073323, Sorkine:2007:ASM:1281991.1282006, Sorkine:2004:LSE:1057432.1057456}), where the user controls the deformation by moving a small number of points on the mesh. These points generate constraints for an optimization problem, whose solution is the deformed mesh. Other common design metaphors are skeleton-based and cage-based. Jacobson et al. \cite{Jacobson:2011:BBW:1964921.1964973} noted the differences between these methods and provided a hybrid method incorporating both. More intricate approaches for mesh deformation use direct control of the mesh normal and curvature instead of vertex positions \cite{journals/cgf/EigensatzSP08,journals/cgf/EigensatzP09}.

\begin{figure}[th]
\centering
\includegraphics[width=1\linewidth]{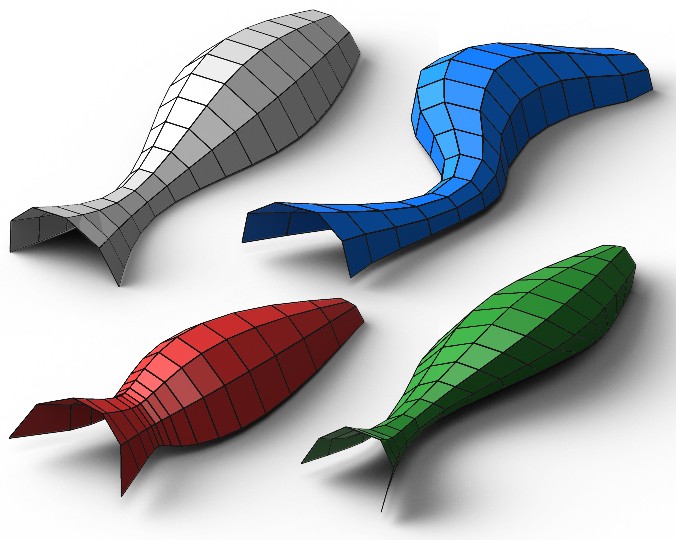}
\caption{Deforming a PM in various linear subspaces. Each result is not achievable in the other two subspaces.}
\label{fig:1}
\end{figure}

Handle-based deformation has also been used in the context of PM's. In \cite{Yang:2011:SSE:2070781.2024158}, the \emph{manifold of polyhedral meshes} was discussed in detail. The idea was to approximate this manifold by an osculant, which is much easier to explore. In this framework, deformation of a PM using positional constraints was made possible; however, computing the osculant is time-consuming and the deformation only \emph{approximately} preserves the planarity of faces. \cite{zhao-2012-idecm} uses the same technique to derive a \emph{curve}-based deformation. 

Recently, Vaxman \cite{Vaxman:2012:MPM:2346796.2346801} described a linear space of PM's by allowing affine transformations per face. In this work, it was proposed to use it instead of the entire manifold, simplifying the math considerably. In fact, this linear space is a special case of the linear spaces to be described in this paper. The main drawback of using this space is its small number of degrees of freedom (dimension). For example, the number of degrees of freedom of a quad PM is about half the size of its boundary, so when the mesh has no boundary, only the trivial, global, transformations are possible (e.g. global rotations). Hexagonal PM's will have only 12 degrees of freedom, regardless of the existence of a boundary. In other words, specifying the geometry of 4 vertices of a PM with hexagonal topology uniquely determines the rest of the PM. Pottmann et al. \cite{Pottmann:2007:GMF:1276377.1276458} described another linear space of PM's, called parallel meshes. It is also a special case of the spaces to be described in this paper.

A second type of mesh deformation is \emph{indirect}. These include various methods that improve the quality of a mesh, such as smoothing and enhancing features. More relevant to us are methods that are used to add variation to a mesh, or to create a collection of meshes based on a single mesh (e.g. \cite{citeulike:2732270}). \cite{Yang:2011:SSE:2070781.2024158}) has also contributed an indirect deformation approach, by designing a user interface which allows to traverse the osculant with ease. In this paper we propose eigenshapes as a way of indirectly adding variation to a PM.

\subsection{Contribution and overview}
In Section 2 we discuss linear subspaces of PM's in detail. We characterize all of the possible maximal subspaces, and show how to construct them. In Section 3 we describe a number of meaningful sets of shapes for editing PM's. In section 4 we discuss practical consideration and limitations.

\section{Linear Subspaces of Polyhedral Meshes}

\textbf{Preliminaries}. In our context, a mesh is defined by a list of vertex geometry and a list of faces. The vertex geometry can be arranged in a $3\times n$ matrix, where $n$ is the number of vertices. We will usually denote this matrix by an upper-case letter, such as $X$ or $Y$. In other words, the vertex geometry is given by
$$X=(x_1,x_2,...,x_n )$$
where the $x_i$ are column 3-vectors. We denote by $F=\{f_j\}_{j=1}^{N_F}$ the set of faces of the mesh, where each face is described as an ordered (oriented) list of vertices. In most cases, $F$ will be common to several meshes, and we will refer to them only by their matrices. Finally, we will denote the vertex coordinates of the face $f$, which is a submatrix of $X$, by $X_f$.

\textbf{Manifolds of Meshes}. When two meshes have the same topology, their linear combination can be defined simply as a linear combination of their vertex geometries. In other words, two meshes $X$ and $Y$ span a linear subspace of meshes defined by
$$\alpha X+\beta Y,\alpha,\beta \in R$$
We can consider the set of all meshes with n vertices and a given topology to be vectors in $R^{3n}$. Obviously, the dimension of this space is $3n$ and is isometric to $R^{3n}$.

Linearly combining two meshes is meaningful because the set of all possible meshes (with a given topology) is a linear space. \emph{PM's}, on the other hand, reside in a complicated, curved submanifold in this space. Linearly combining two PM's will usually not result in a PM, which is the cause of many of the problems in using them. It so happens that the manifold of PM's may be \emph{covered} by linear submanifolds, which we discuss next. By replacing the non-linear constraints defining the manifold of PM by linear ones, many of the problems related to PM design disappear. We emphasize two important points. First, the dimensions of the linear subspaces are much less than $3n$. Hence, making sure that they are \emph{maximal} is crucial. Second, the set of linear subspaces must \emph{cover} the submanifold of PM's in such a way that any PM will be reachable from any other PM by moving through those linear spaces alone. 

\textbf{Linear Spaces}. To investigate the linear subspaces of PM's we first examine the linear subspaces of much simpler entities: planar polygons; they are simply PMs with a single face, so we can use the same notation. We will further simplify the discussion by assuming that the polygons are geometrically centered at the origin. This is not detrimental to the generality of the arguments, since centering is a linear operation. In addition, we will assume that the polygons are not degenerate. While degenerate polygons have a place in this theory, they do not appear in practice, and therefore cause an unnecessary complication.

\begin{theoa}
Let $X$ and $Y$ be two $3\times k$ matrices representing the geometry of two planar k-gons in $R^3$ with $k > 3$ vertices, centered at the origin. Let $N_X$ and $N_Y$ be the unit $(1\times 3)$ row vectors normal to the planes defined by $X$ and $Y$ respectively. Then $X$ and $Y$ span a linear subspace of planar polygons iff at least one of the following holds:
\end{theoa}

\emph{Relationship of type 1}: $X$ is an affine transformation of $Y: X=AY$ for some $3\times 3$ matrix $A$.

\emph{Relationship of type 2}: $N_Y X=cN_X Y$ for some scalar $c$.

\begin{proof}: By definition,
\begin{equation} \label{eq:1}
    N_X X=N_Y Y=0
\end{equation}
where 0 is a zero k-vector. Assume that $X$ and $Y$ span a linear subspace of planar polygons. This implies that every linear combination of them defines a plane and thus has a normal vector. In other words, for each $\alpha$ and $\beta$, there exists a vector $N(\alpha ,\beta)$ such that
\begin{equation} \label{eq:2}
N(\alpha ,\beta )(\alpha X+\beta Y)=0
\end{equation}
Consider the set of vectors ${N(\alpha ,\beta ),\alpha ,\beta \in R}$. First, assume that it has dimension 3. Then there exist $N_i=N(\alpha _i,\beta _i ), i=1,2,3$ such that the $N_i$'s are not collinear. By (\ref{eq:3}), we can write
$$(\alpha _i N_i )X=-(\beta _i N_i )Y$$
where $(\alpha _i N_i )$ is a $3\times 3$ matrix whose rows are $\alpha_i N_i$. $(\beta_i N_i)$ is defined similarly. Since the $N_i$'s are independent, we can invert $(\alpha _i N_i ): X=-(\alpha _i N_i )^{-1} (\beta _i N_i )Y$. Hence, $X$ is an affine transformation of $Y$, which is type 1.

If the set of normals has dimension less than 3, then this set must be spanned by $N_X$ and $N_Y$. Thus, we can write (\ref{eq:2}) as
$$(A(\alpha ,\beta ) N_X+B(\alpha ,\beta ) N_Y )(\alpha X+\beta Y)=0$$
Expanding the LHS and using (\ref{eq:1}) we obtain
$$N_Y X=-\frac{A\beta}{B\alpha} N_X Y$$
which, noting that  $\frac{A\beta}{B\alpha}$  is a constant, is the relationship of type 2.
In the other direction, first, if $X$ and $Y$ are planar and $X$ is an affine transformation of $Y$ (type 1 relationship) then the rank of each of the matrices $X$ and $Y$ is 2 and there exists a $3\times 3$ matrix $A$ such that $X=AY$. Hence, their combination
$$\alpha X+\beta Y=(\alpha A+\beta I)Y$$
has $rank \le 2$ and thus is planar.

Second, if $X$ and $Y$ are planar and $N_Y X=cN_X Y$ for some scalar $c$ (type 2 relationship), then for any scalars $\alpha$ and $\beta$ we can choose $A$ and $B$ such that $c=\frac{A\beta} {B\alpha}$. Working our way backwards, this implies that
$$(AN_X+BN_Y )(\alpha X+\beta Y)=0$$
concluding that $X$ and $Y$ spans a linear subspace of planar polygons. 
\end{proof}

The following corollaries follow immediately:

\textbf{Corollary 1}: If $X$ and $Y$ are parallel planar polygons, then they span a linear subspace of planar polygons.

\textbf{Corollary 2}: if $X$ and $Y$ are planar polygons not related by any affine transformation, yet they span a linear space of planar polygons, then the normals of their linear combinations must be linear combinations of their normals.

Using Theorem 1, it is simple to prove an analogous result for PM's:

\begin{theoa}
Let $X$ and $Y$ be two PM's in $R^3$ with common topology. Then $X$ and $Y$ span a linear space of PM's iff each non-triangular face of $X$ has a type 1 or type 2 relationship with the corresponding face of $Y$. 
\end{theoa}

\begin{figure}[th]
\centering
\includegraphics[width=1\linewidth]{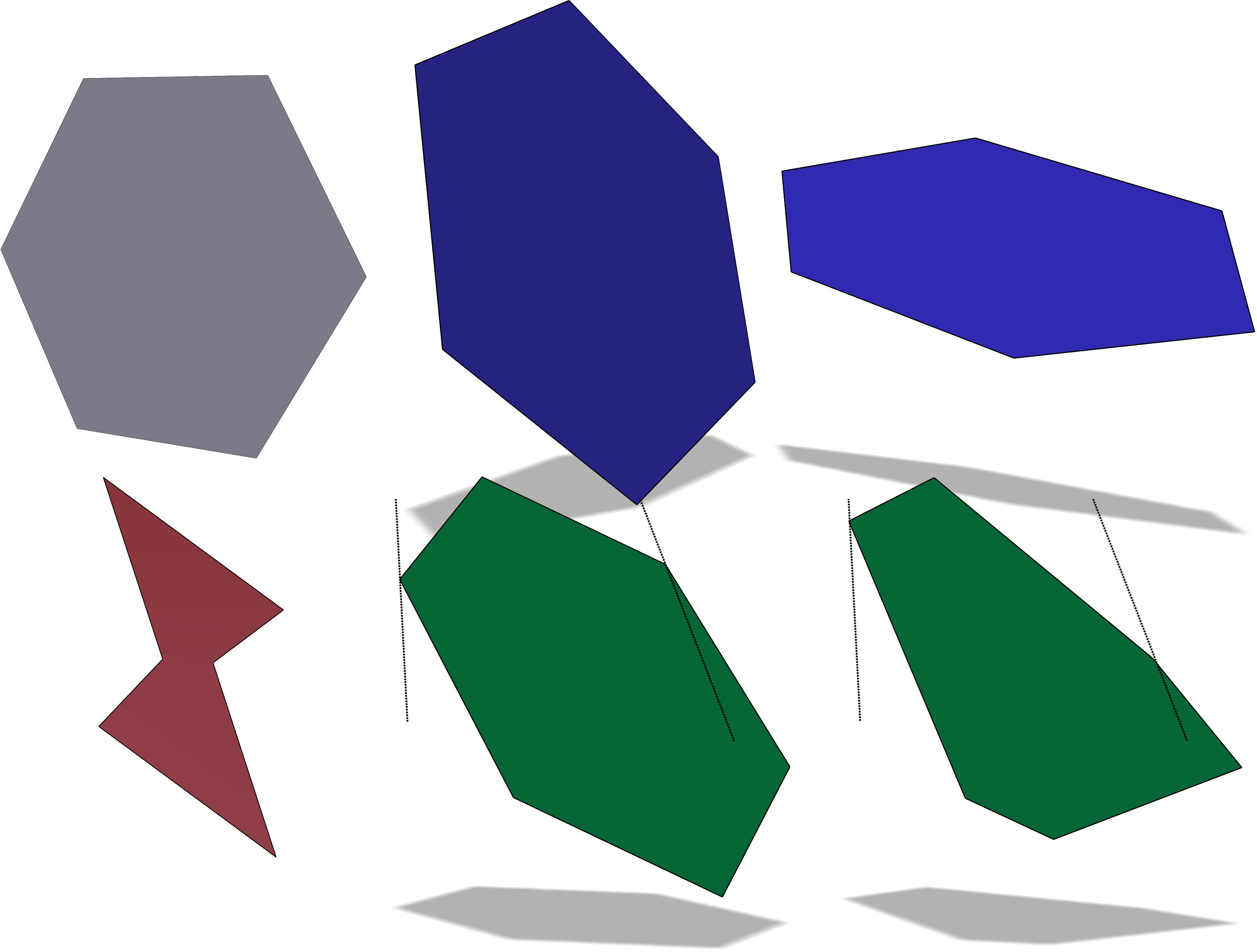}
\caption{Illustrating polygon relationship types. The blue hexagons are related to the gray one (top left) by an affine transformation, hence of type 1. The red hexagon is parallel to the gray one, which is a special case of type 2. The two green hexagons have identical normals and a type 2 relationship to the gray one. The relationship is maintained as long as the vertices of the green polygons slide on the dotted lines, which are equidistant from the plane of the grey hexagon.}
\label{fig:3}
\end{figure}

\textbf{Generating subspaces}. Our goal is to explore the linear subspaces of PM's which contain a given PM, but first we need to characterize those linear subspaces. We say that a PM $Y$ \emph{generates} a linear subspace if that subspace contains $Y$. Given a linear subspace of PM's $V$ generated by $Y$, Theorem 2 tells us that the relationship of each face of any $X\in V$ to the corresponding face of $Y$ must be either of type 1 or type 2. Therefore, we can generate a subspace $V$ from $Y$ by first specifying, for each face of Y, which type of relationship applies, and then finding all possible $X$ that are related to $Y$ in that way.

In practice, we can treat each face $Y_f$ of $Y$ as a separate planar polygon, and write a linear system for $X_f$,
\begin{equation} \label{eq:3}
B_f vec(X_f )=0
\end{equation}
where the matrix $B_f$ depends on $Y_f$ and the type of relationship, and $vec(X_f)$ is a vectorization of the matrix $X_f$. We can combine all of the small linear systems into one large linear system
$$Bvec(X)=0$$
and the solution space of this system, namely, the nullspace of $B$, is exactly $V$. The matrix $L_B=B^T B$ may also be used since they share the same null space. We remark that $L_B$ can be normalized to give a Laplacian-like operator. This operator may be related to yet another differential operator described in Appendix 1, but we have yet to pursue this connection.

To construct $B_f$ we consider each type of relationship separately. A relationship of type 1 means that $X_f$ must be an affine transformation of $Y_f$. Thus $X_f$ must satisfy
$$X_f (Y_f^+ Y_f-I)=0$$
where $Y_f^+$ is the pseudo-inverse of $Y_f$. This equation can easily be transformed into the form (\ref{eq:3}). PM's related in this manner were explored by Vaxman \cite{Vaxman:2012:MPM:2346796.2346801}.

As for a relationship of type 2, we start by examining its definition, $N_{Y_f} X_f=cN_{X_f}Y_f$. Given $Y_f$, the normal $N_{Y_f}$ is also given and so we are left with $X_f$, $N_{X_f}$ and $c$ as variables. 
Based on $N_{X_f}$, we can further divide this into 2 sub-cases:
\begin{enumerate}
\item    $N_{X_f}=N_{Y_f}$.
    In this case, we simply have,
    $$N_{Y_f} X_f=N_{X_f} X_f=0$$

\item    $N_{X_f}\ne N_{Y_f}$.
    Note that in this case the (column) vector $N=N_{X_f}\times N_{Y_f}$ is in the intersection of the planes of $X$ and $Y$. We have the following decomposition for $X_f$ and $Y_f$,
    $$X_f=NX_1+E_X X_2, \quad Y_f=NY_1+E_Y Y_2$$
    where $E_X=N\times N_X$ and $E_Y=N\times N_Y$, and $X_i$,$Y_i$,$i=1,2$ are row vectors with an element for each vertex. Thus $N_Y X_f=N_Y E_X X_2$ and $cN_X Y_f=cN_X E_Y Y_2$, and the relationship of type 2 implies
    $$X_2=\frac{cN_X E_Y}{N_Y E_X} Y_2=c' Y_2$$
    which means that $X_f=NX_1+c' E_X Y_2$. We apply the cross product by $N$ to both sides to get
    $$X_f\times N=c' E_X Y_2\times N$$
    Let $M$ be the null space of $Y_2$, i.e. $Y_2 M=0$. We finally have that $(X_f\times N)M=0$, which can also be written in the form (\ref{eq:3}). This means that in order to generate a linear subspace based on a relationship of type 2, we must first prescribe the normal $N_X$ or equivalently, $N$.
\end{enumerate}

Thus given a planar polygon, we can generate 3 types of linear subspaces of planar polygons. We now proceed to prove that each type of space generated for a face is maximal. Again we start with the simpler case of planar polygons.

\begin{theoa}
Let $V=null(B)$ be a linear subspace generated by a (non-degenerate) polygon $X$, where $B$ is constructed as described above. Then $V$ is a maximal linear subspace of planar polygons.
\end{theoa}

\begin{proof}
First, assume that $V$ is the space of all affine transformations of $X$. Let $Y$ be a planar polygon such that $Y\notin V$. Then $X$ and $Y$ must have a relationship of type 2: $N_Y X=cN_X Y$. We can assume w.l.o.g that $N_X$  and $N_Y$ are not collinear. Let $R$ be a rotation matrix around $N_x$. Then $RX\in V$ and $N_Y RX=c' N_X Y$. This implies that $N_Y RX=c'' N_Y X$, or $N_Y (R-c'' I)X=0$. Hence $N_Y$ and $(R-c'' I) N_X$ are collinear.  However,  $(R-c'' I) N_X$ and $N_X$ are also collinear and this in turn means that $N_X$ and $N_Y$ \emph{are} collinear. Thus we have a contradiction.

The second part of the proof is subdivided to two cases. First, let $V$ be the space of all polygons which are \emph{parallel} to $X$ and define $Y$ similarly. $Y$ cannot be related to all polygons in $V$ by an affine transformation, so assume w.l.o.g. that $X$ and $Y$ have the relationship of type 2. Then applying the same rotation of $X$ strategy used previously we reach contradiction again.

Finally, we consider the case where $V$ is the space of all polygons with type 2 relationship to $X$. To define this space we need to set another vector $N$ in the plane of $X$, which is shared among the planes of all polygons in this space. Again, $Y$ has w.l.o.g. a relations ship of type 2 to $X$. $Y$ cannot contain $N$ since by construction it would mean that $Y\in V$. This means that there is another vector $N'$ that the planes of $X$ and $Y$ share. The vertices of $X$ are free to move in the direction of $N$, and the new polygon $X'$ will still be in $V$. $X'$ cannot hold a relationship of type 2 with $Y$ (since $N$ and $N'$ are different) and $X'$ can be chosen so it will not be an affine transformation of $Y$, and we reach contradiction yet again.
\end{proof}

In reality, to avoid having to specify an explicit normal for each face having a relationship of type 2, we used three cases when specifying relationships types for faces. The first case, which we call the \emph{affine case}, is simply when all faces have type 1 relationship. In the second case, the target face normal was chosen to be identical to the source normal. The subspace generated by this case is that of all polygons which are parallel to the source polygons, hence, the \emph{parallel case}. In the third case, the \emph{vertical case}, all the face normals were set to the \emph{up} (vertical) vector. The justification for this is the fact that many meshes, especially architectural meshes, have a prominent up direction.

Theorem 3 tells us that the construction of linear subspaces of \emph{PM's} can produce all the maximal linear subspaces. However, some linear subspaces constructed that way may not be maximal. This can happen when the constraints imposed on a face by other faces, and its own linear subspace, cause it to be in another linear subspace. For example, the red cube in Fig. \ref{fig:4} was deformed in the parallel subspace. In this space, each face can only be stretched in the obvious directions, which is a subset of the affine transformations of the face. Thus, the parallel subspace in that case is not maximal, since it is contained in the affine subspace. Note that by removing a single face from the cube, the linear subspaces become different. These situations are easily detectable, and the face can be reassigned. However, in our experiments we rarely encountered such situations.

It is now easy to show that there is a piecewise linear path between any two PM's in the manifold: using the affine space generated by the two meshes, they can be projected to the same plane, where they share the parallel space. This construction however is not very useful as it does not provide any insight into the manifold itself. Nevertheless, it forms a loose "lower bound".

In our examples, the relationship types per face were color coded by blue, red and green for the affine, parallel and vertical cases, respectively. When more than a single relationship type is used to generate the subspace, it is referred to as a \emph{mixed space}. We note here that in the case of the affine and vertical spaces, $B$ can be separated to three identical matrices, operating on $x$,$y$ and $z$ separately. Exploiting this, the performance of the algorithms presented in the next section can be significantly improved.

\begin{figure}[th]
\centering
\includegraphics[width=1\linewidth]{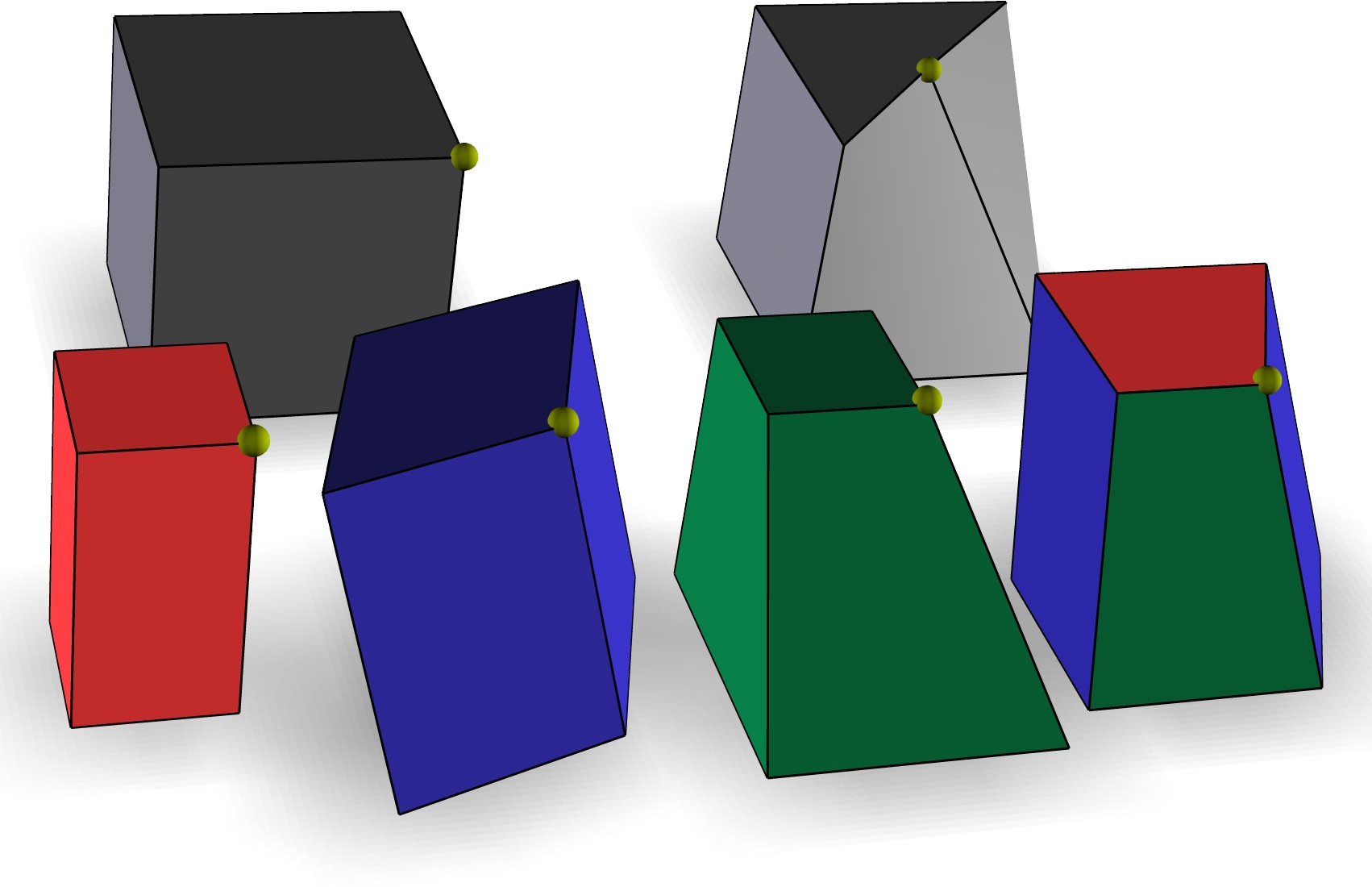}
\caption{Hexahedron in different subspaces generated by the (gray) cube on the top left. They are the closest ones in their subspaces to the (gray) non-PM on the top right, subject to the hard constraint imposed by the yellow vertex.}
\label{fig:4}
\end{figure}

\textbf{Degrees of freedom}. The number of degrees of freedom (NDOF) of a linear subspace of PM's is exactly the dimension of the nullspace of $B$. We can estimate the NDOF in some specific cases, such as when the space is not mixed. The NDOF is then exactly the co-rank of $B$. However, this value depends too much on the current embedding of the PM and does not give any insight into the relation to its topology. We instead provide a lower bound on the NDOF for a given PM, which can be inferred from the topology alone.

Denote by $N_v$,$N_b$,$N_e$,$N_f$,$N_c$ the number of vertices, boundary vertices, edges, faces and corners of the PM, respectively. The number of variables (the mesh vertex geometry) is always $3N_v$. In the affine case, the number of equations is $3N_c$, but each face is determined by just three vertices. Hence a lower bound on the NDOF is $3(N_v+3N_f-N_c)$. Similarly, in the parallel case the lower bound is $3N_v-N_c+N_f$, and in the vertical case it is $3N_v-2(N_c-2N_f)$.

We can use the generalized Euler formula, $N_v-N_e+N_f-b=2g$, where $b$ is the number of boundaries, and $g$ is the genus of the mesh, and the fact that $N_c=2N_e-N_b$ to obtain
$$N_c=2(N_v-2g+N_f-b)-N_b$$
Plugging this into the formulas for the NDOF yields an expression that does not depend on $N_c$ and $N_e$. For (semi-) regular graphs, $N_f$ can also be expressed using $N_v$ and $N_b$ and vice-versa, which may give more intuitive results. Additionally, we define the \emph{number of free vertices} (NFV) as the NDOF divided by 3. The NFV roughly gives the number of vertices that can be fixed independently. We list the minimal NVF for quad and hex meshes for both cases in Table 1.

The table shows that the minimal NFV for quad meshes in the affine and parallel cases is determined by the size of the boundary. In fact, our experiments show that, apart from very symmetric cases like spheres or tori, the minimal NFV for the affine case is the true NFV, up to a global transformation. This means that there is very little that can be done with closed quad meshes in the affine case. The situation is even worse for hex meshes: unless the mesh is just a strip of hexagons, the minimal NFV will be negative. 
\begin{wrapfigure}{r}{0.13\textwidth}
  \begin{center}
    \includegraphics[width=0.12\textwidth]{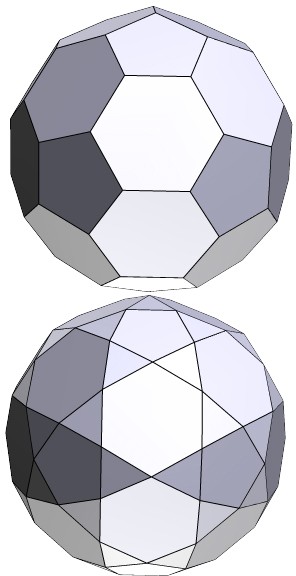}
  \end{center}
\end{wrapfigure}
In fact, we prove in Appendix 2 that the actual NFV is 3 for any 3-regular mesh without a boundary. A trick that can be used to increase the NFV is to apply a half-edge subdivision to the hex mesh (see inset). Technically, the new mesh will not be a hex mesh, but it might retain the "look" of the original hex mesh, and the minimal NFV will be much higher. As for the parallel case, it is easy to show that for closed 3-regular PM's, the NFV is exactly $N_f$. 

\begin{table}
\centering
\caption{Minimal number of free vertices (NFV) in different subspaces}
\label{tab:NFV}
\begin{tabular}{c c c}
\hline\noalign{\smallskip}
         &   Quad mesh & Hex mesh \\
\hline\noalign{\smallskip}
Affine   & $\frac{N_b}{2}+b+2g$ & $\frac{-N_v}{2}+\frac{3}{4}N_b+\frac{3}{2}b+3g$ \\
\hline\noalign{\smallskip}
Parallel & $\frac{N_b}{2}+b+2g$ & $\frac{N_v}{6}+\frac{5N_b}{12}+\frac{5b}{6}+\frac{5g}{3}$ \\
\hline\noalign{\smallskip}
Vertical & \multicolumn{2}{c}{$-\frac{N_v}{3}+\frac{2N_b}{3}+\frac{4b}{3}+8g/3$} \\
\hline\noalign{\smallskip}
\end{tabular}
\end{table}

\section{Exploring Linear Subspaces}
\textbf{Overview}. Once all faces of the mesh have relationship types assigned to them and the matrix $B$ is computed, we can begin the exploration of $null(B)$. While we can do this by simply computing an orthogonal basis for $null(B)$, it may not be very useful: this basis will contain random PM's. Instead, we discuss ways to create more meaningful shapes, which are targeted toward different levels of editing.

\textbf{Eigenshapes}. \cite{Yang:2011:SSE:2070781.2024158} proposed to explore the manifold of PM's not by explicitly setting positional constraints, but by traversing the neighborhood of the PM. This is done by choosing a few directions (two or three for easy navigation) on the osculant which match the manifold the best. Using linear subspaces, we do not have to worry about going far away from the manifold, which allows us to be more adventurous with the exploration. We propose using the PM's "harmonics" as a basis for exploration. More precisely, we use the eigenvectors of the Laplacian $L$ constrained to the linear space, which we call eigenshapes. These are defined by the constrained Rayleigh quotient:
\begin{equation} \label{eq:4}
\max_X\frac{X^TLX}{X^TX}  \quad s.t.  BX=0
\end{equation}

The solution to this problem is found in \cite{Gander1989815} as the eigenvectors of $PLP$ where $P=I-B^T (BB^T )^{-1} B$. See implementation details in Section 4 on how to compute the eigenshapes efficiently. To effectively visualize the eigenshapes and to explore them efficiently, we suggest the following idea: add the eigenshapes to the source PM and apply a "band-pass-filter" to it. By sliding the filter we can quickly see how eigenshapes of different frequencies affect the PM (Fig. \ref{fig:11})

\textbf{Sparse shapes}. Habbecke and Kobelt \cite{Habbecke:2012:LAN:2322116.2322124} discussed editing of constrained meshes, where their goal was to be able to reposition a vertex while making as little as possible change to the rest of the mesh and satisfy the constraints. This addresses the well-known problem of editing with constraints, where making a change in one portion of a mesh damages the work that was already done elsewhere in the mesh. Their approach is based on linearizing the constraints and finding sparse solutions to the linearized system. The same strategy can be used to deform PM's and in fact, one of the constraints treated in \cite{Yang:2011:SSE:2070781.2024158} is the planarity of faces. In terms of basic shapes, in order to be able to move just a small set of vertices, a shape where most of the vertices lie on the origin is needed. These \emph{sparse shapes} are just sparse vectors in $null(B)$. To find sparse solutions, Habbecke and Kobelt employ the Orthogonal Matching Pursuit (OMP) algorithm \cite{omp}, and the same can be done to find sparse shapes.

For many subspaces, the only sparse shapes that can be found are not sparse at all. For example, the affine space for quad meshes contains truly sparse shapes only for very symmetric cases (Fig. \ref{fig:5}). In these cases \emph{approximate} sparse shapes - shapes that are not in the linear subspace but close to it - can be found instead. For comparison, the accurate sparse shape in the middle of Fig. \ref{fig:6} has $||BX|| \approx 10^{-12}$, and the approximate sparse shape has $||BX|| \approx 10^{-4}$. The original PM was produced by planarizing a deformed torus, which had $||BX||\approx 0.1$.

\begin{figure}[th]
\centering
\begin{tabular}{ccc}
    Affine & Dual & Vertical \\
    \includegraphics[width=0.3\linewidth]{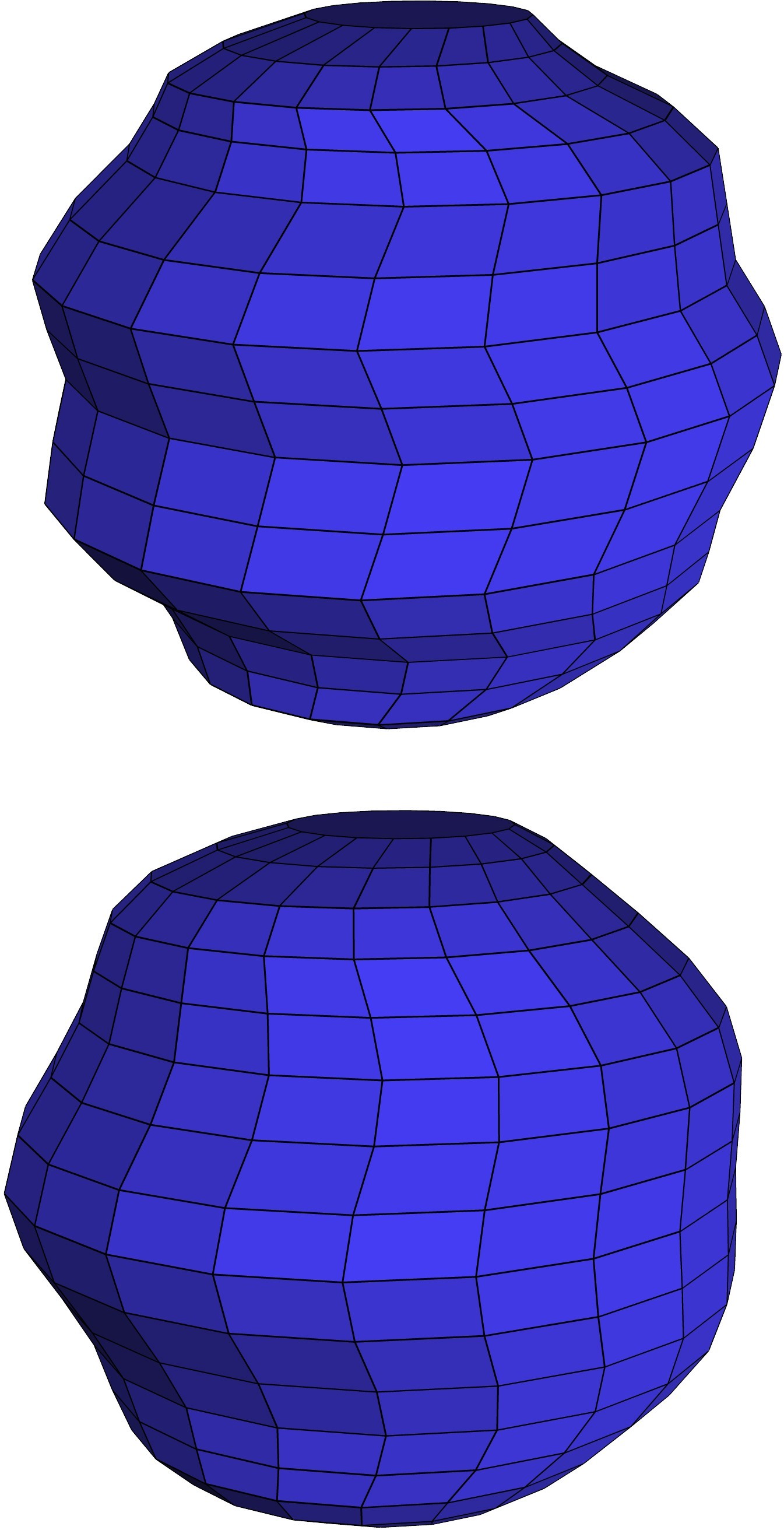} &
    \includegraphics[width=0.28\linewidth]{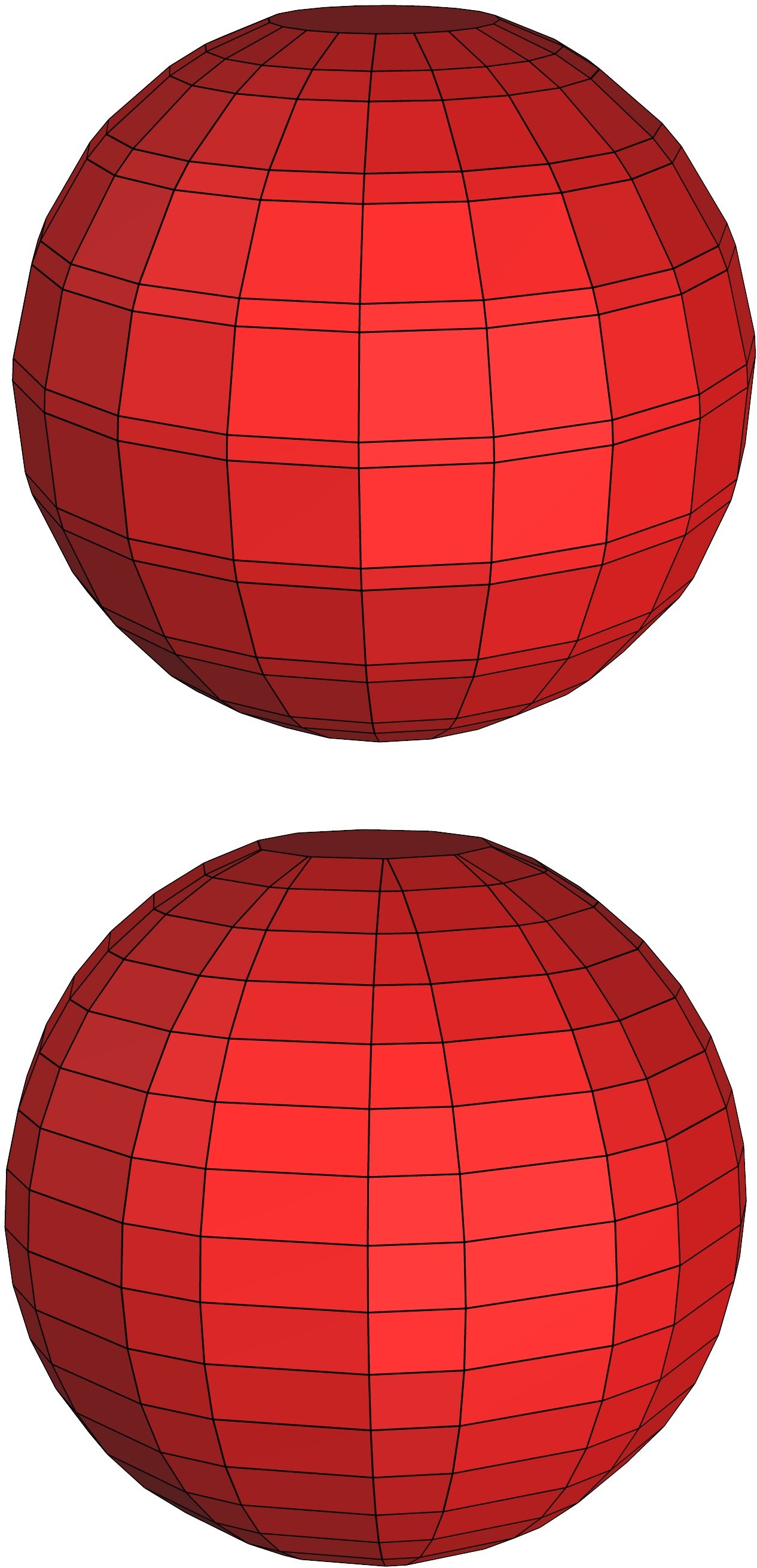} &
    \includegraphics[width=0.3\linewidth]{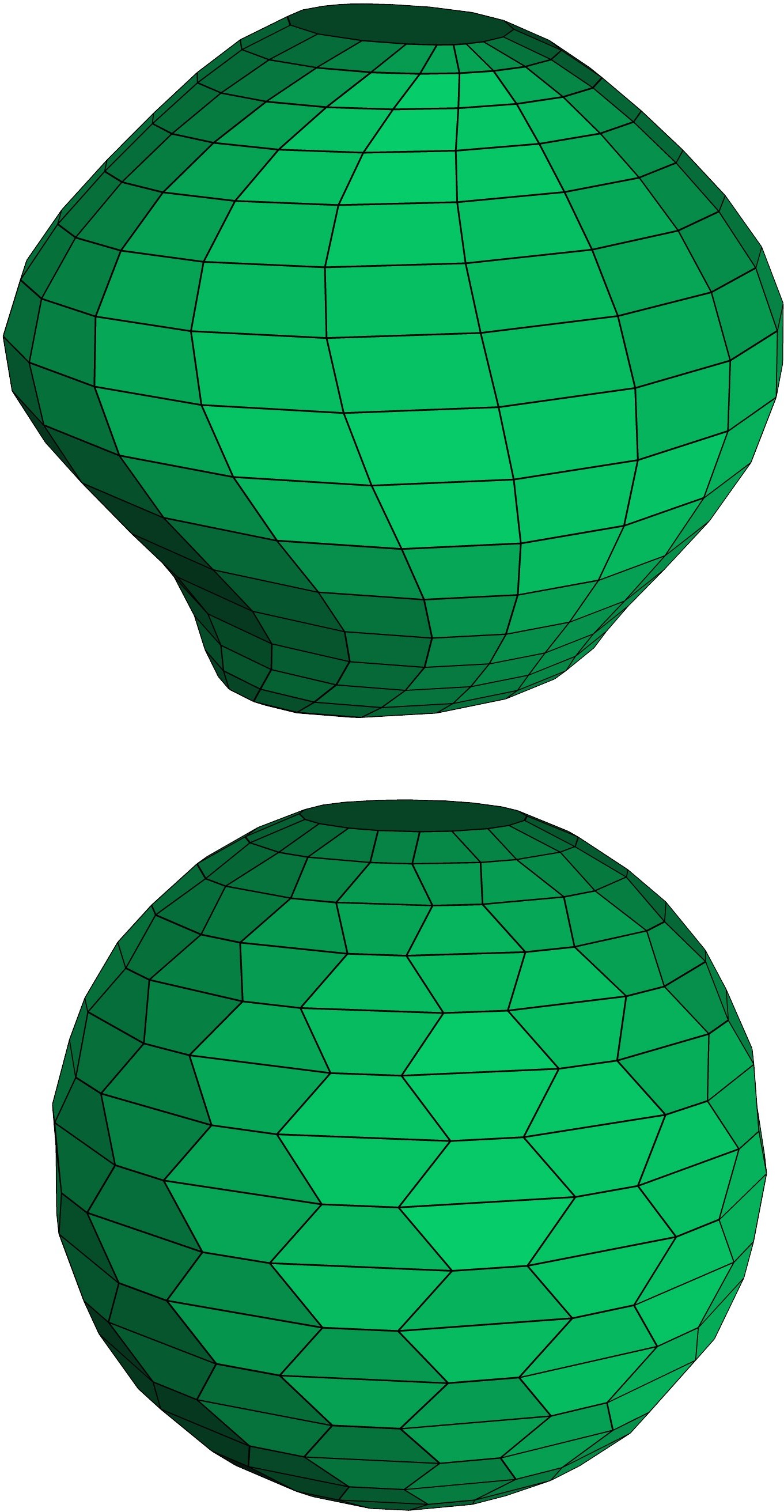}
\end{tabular}
\caption{Adding eigenshapes of different subspaces to a simple spherical quad PM. See also accompanying video.}
\label{fig:5}
\end{figure}

\begin{figure}[th]
\centering
\includegraphics[width=1\linewidth]{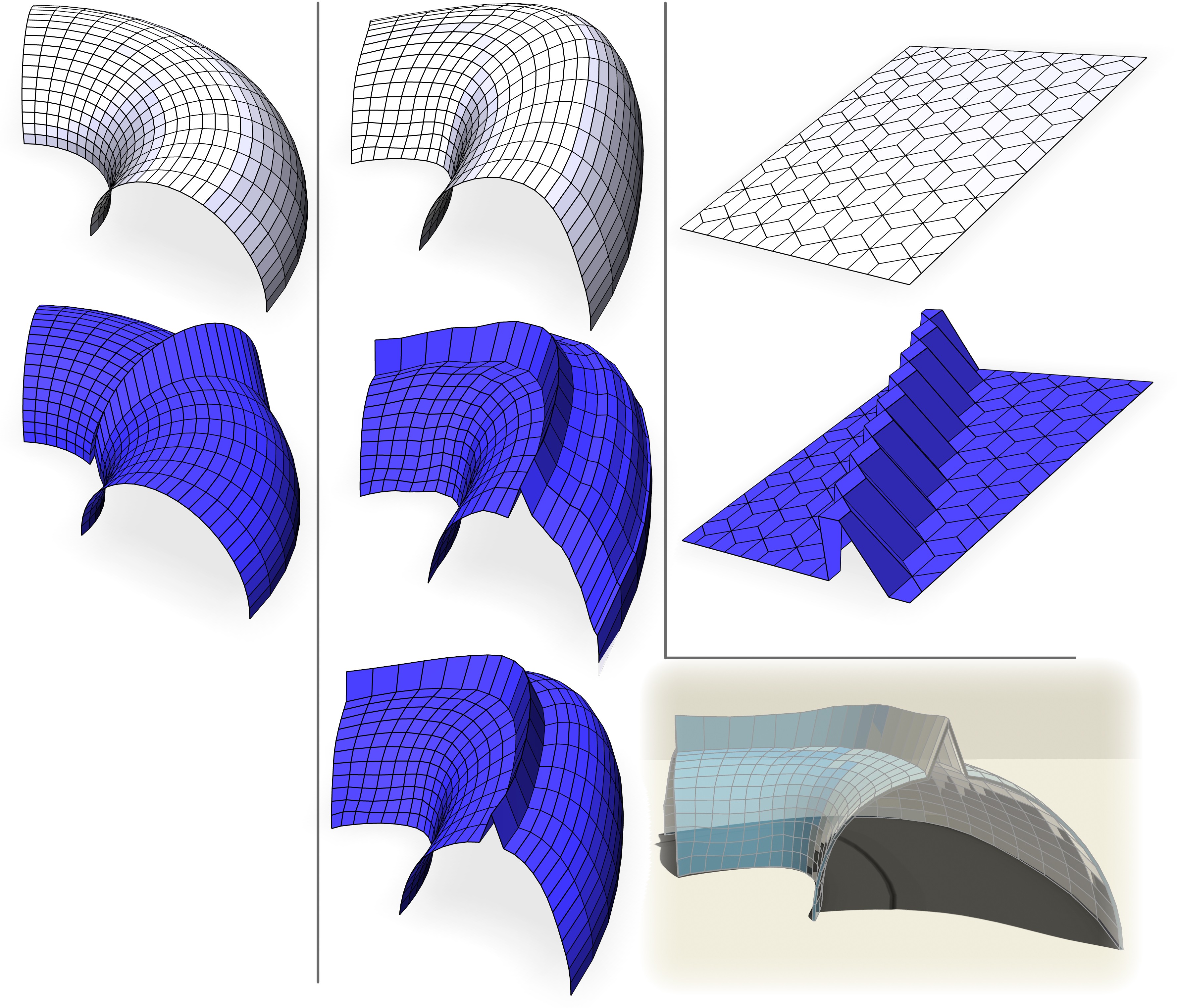}
\caption{Sparse shapes. (Left) Part of a symmetric torus quad PM, having an accurate sparse shape. (Middle) Deformed torus. Its accurate sparse shape is not sparse at all, but it has an inaccurate sparse shape. (Right) Sparse shape of a flat PM.}
\label{fig:6}
\end{figure}

\textbf{Fundamental shapes}. While a sparse shape changes only a small number of vertices, it can still be non-local, moving vertices on opposite sides of the PM. In many cases a shape with more locality is required; one that perhaps moves \emph{all} vertices, but to a lesser extent. To elaborate, suppose a vertex $v_i$ has been selected. We may then define the fundamental shape associated with $v_i$ as the solution to the optimization problem
$$
\min_X || X-\delta_i ||^2 + \lambda ||LX||^2 \\
s.t.  BX=0$$
where $\delta_i$ is a vector whose only non-zero elements are the ones corresponding to $v_i$ and $LX$ is a regularization term. Of course, both the distance function and the regularization terms can be replaced by other similar functions.

\textbf{Handle-based deformation}. PM's can be deformed directly, and the handle-based approach is probably the most natural metaphor to use (excluding, perhaps, the recent curve-based approach \cite{zhao-2012-idecm}). This was studied in detail in \cite{Vaxman:2012:MPM:2346796.2346801} and \cite{planarization} for the case of PM's in the affine case only, where an As-Rigid/Similar-As-Possible (ARAP/ASAP) deformation was computed within the resulting subspace. The well-known solution to the ARAP/ASAP deformation problem uses an alternating local/global scheme \cite{Liu2,Sorkine:2007:ASM:1281991.1282006}. The only difference when applying this to PM's is that the constraints defining the linear subspace must be satisfied when solving the global steps. In Fig. \ref{fig:7} we used the same method as in  \cite{Vaxman:2012:MPM:2346796.2346801} to deform in an ASAP way a half sphere hex mesh in the non-mixed spaces. The boundary was kept fixed and one vertex on the top was moved slightly higher. The affine subspace allows only global transformations and the parallel subspace produced self-intersections almost immediately. The vertical subspace produced pleasing, nontrivial results.

\begin{figure}[th]
\centering
\includegraphics[width=1\linewidth]{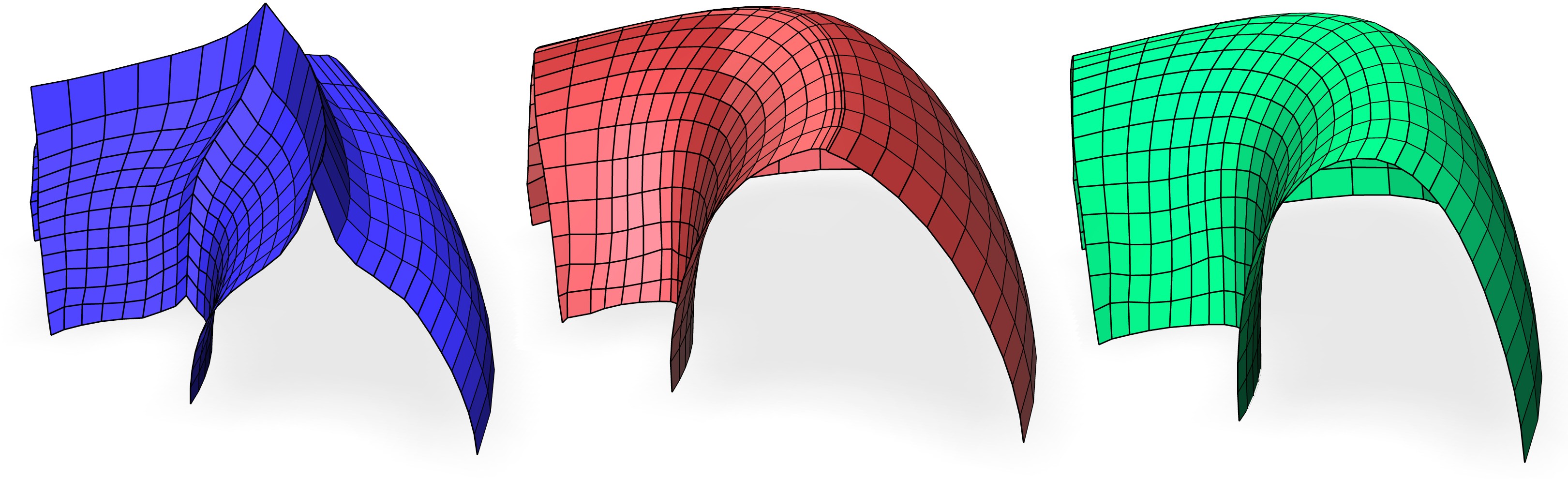}
\caption{Fundamental shapes of the deformed torus.}
\label{fig:7}
\end{figure}

\begin{figure}[th]
\centering
\includegraphics[width=1\linewidth]{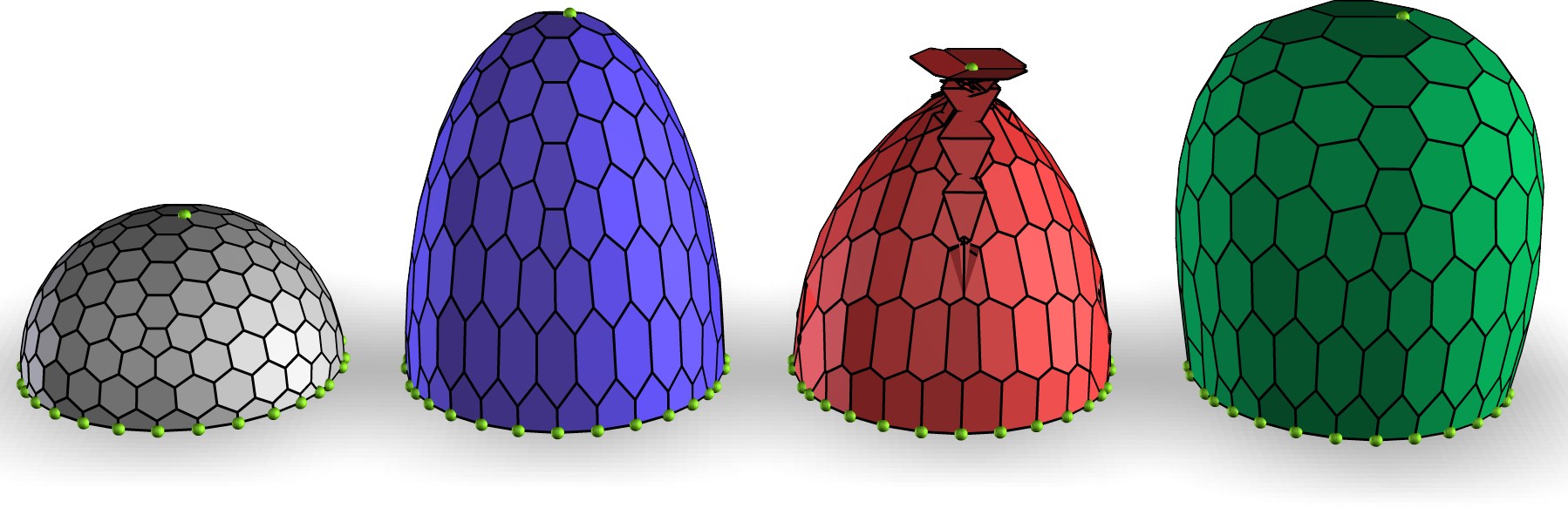}
\caption{ASAP deformation of a hexagonal half sphere. Note that in the (blue) affine subspace, only global transformations are possible.}
\label{fig:8}
\end{figure}

\textbf{Dual exploration}. Every polyhedron admits a family of dual polyhedra, most notably the polar dual \cite{richter1996realization}, having the property that the vector to each of the dual vertices is parallel to the corresponding primal face. Usually polar duals are associated only with star-shaped polyhedrons, since otherwise the polar dual may self-intersect. Here we ignore this and associate polar duals with general, non-convex PM's. Obviously the polar dual associated with a PM is itself a PM, so the ideas presented in this paper apply also to the space of polar duals to a given PM. This essentially means that we can explore the subspace of the PM based on its \emph{face normals} instead of the vertex positions. Although the subspaces defined using the face normals are linear, since they are the same as the linear spaces of polar duals, they are not linear with respect to the vertices of the primal mesh. The reason is that the duality transformation is not linear. Still, it involves only solving a sparse linear system and can be done in real time.

The benefit of dual exploration of PM subspaces is that this gives a completely different number of DOFs compared to the primal space, based on the normal of the faces instead of the vertices. As an extreme example, the duals of any 3-regular meshes are triangle meshes, which trivially preserve planarity. Hence, editing a 3-regular mesh in the normal domain is also trivial: any choice of normal will result in a valid PM. Fig. \ref{fig:9} shows the dual deformation of two PM's. The results there could not have been achieved using only one primal linear space.

\begin{figure}[th]
\centering
\begin{tabular}{ccc}
    Primal & Dual & Deformed Primal \\
    \includegraphics[width=0.28\linewidth]{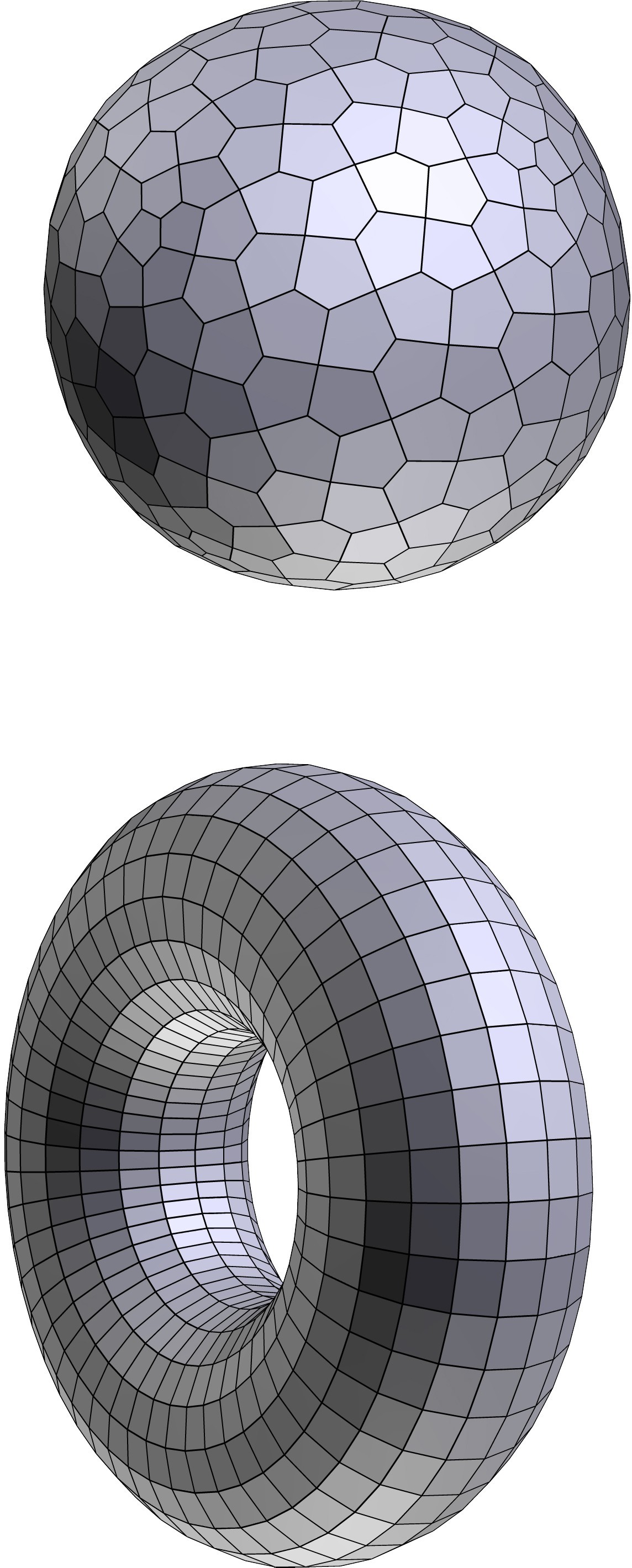} &
    \includegraphics[width=0.28\linewidth]{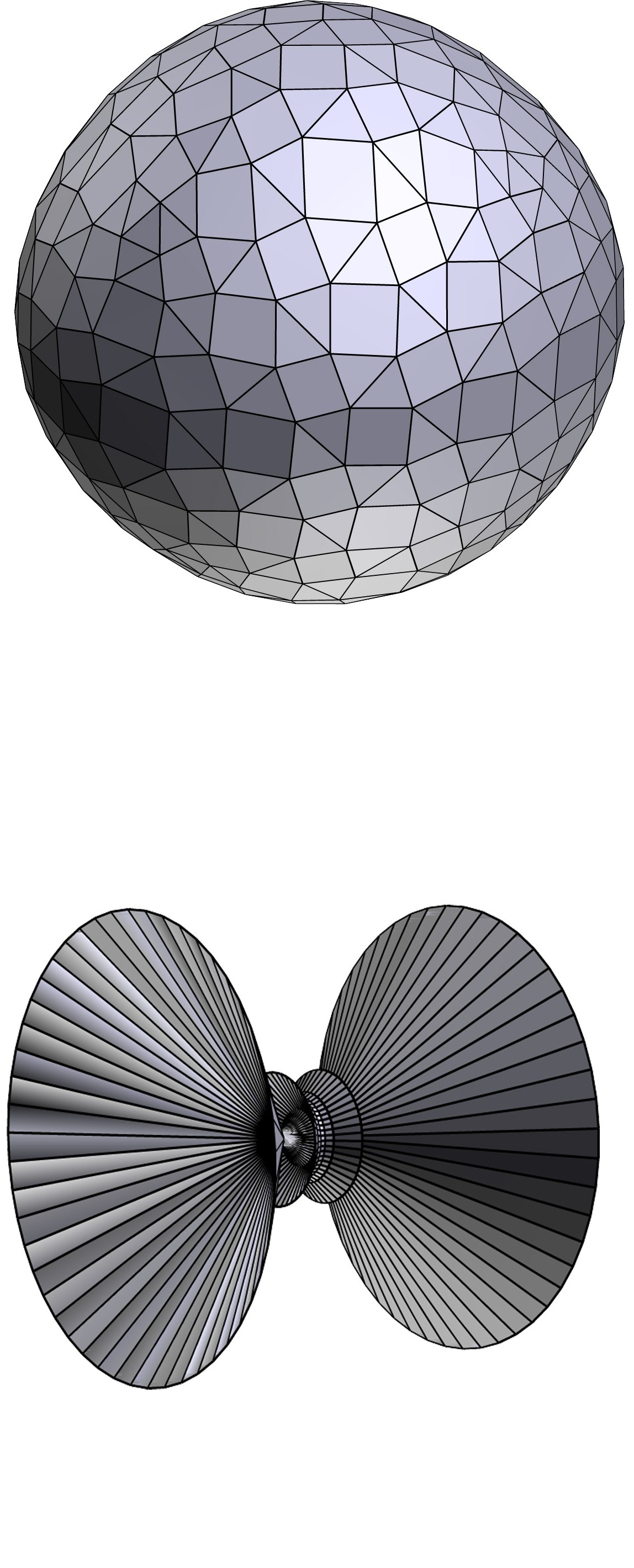} &
    \includegraphics[width=0.28\linewidth]{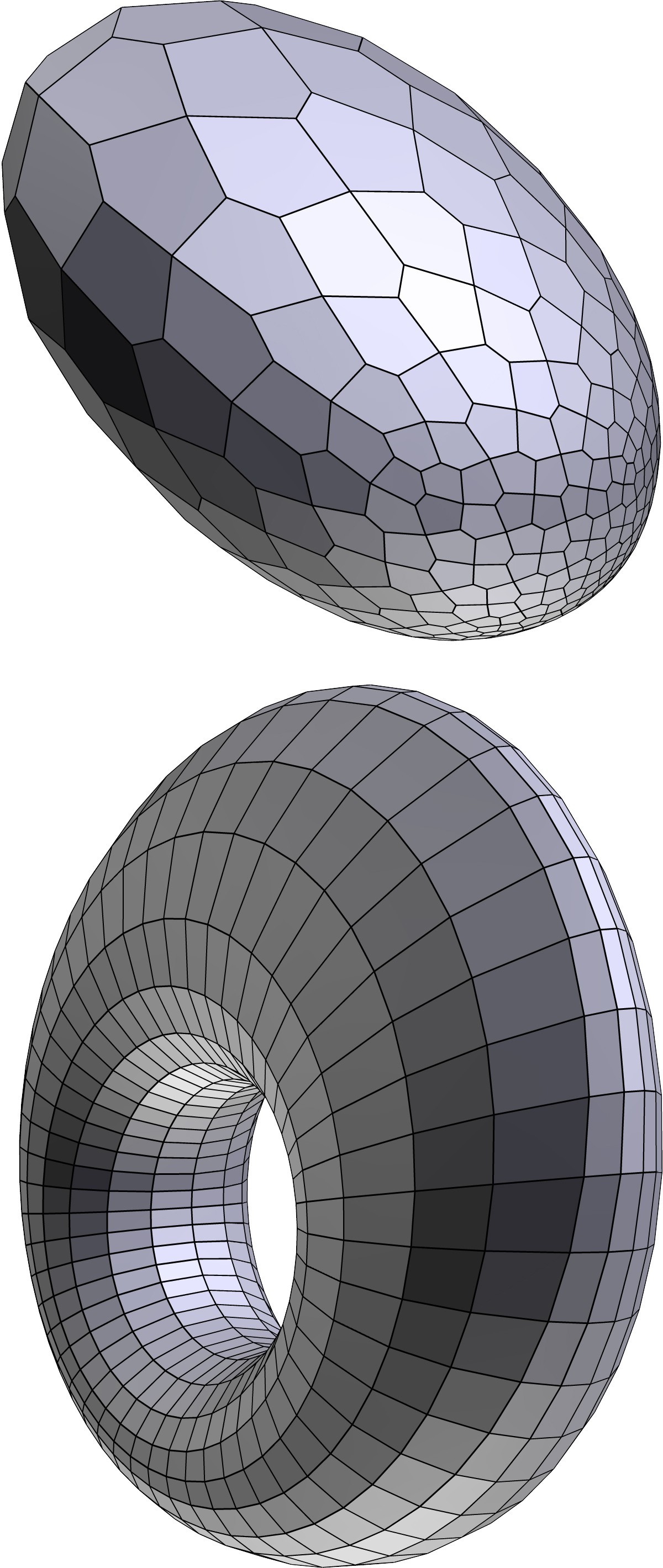}
\end{tabular}
\caption{Deformation of a  (left) sphere and a torus using the (middle) polar dual. In both cases an eigenshape of low frequency was added to the dual mesh, and a new (right) primal mesh results.}
\label{fig:9}
\end{figure}

\section{Discussion and Future Work}
\textbf{Implementation details}. Most of the software implementation was done in MATLAB, and was wrapped as a plugin for Autodesk Maya, for its user interface. The matrix $B$ was built by constructing $B_f$ face-by-face. $B_f$ as defined here is already not full rank, so we reduced the number of equations per-face using SVD. The construction takes less than a second for meshes with approximately a thousand faces.

To compute the eigenshapes, a sparse QR decomposition was used to generate an orthonormal basis $N$ of $null(B)$, then any $X$ in $null(B)$ can be written as $NW$ for some $W$, and 
$$\max_X \frac{X^T LX}{X^T X}=\max_X \frac{W^T N^T LNW}{W^T N^T NW} = \max_X \frac{W^T N^T LNW}{W^T W}$$
which is solved using the eigendecomposition of $N^T LN$. This approach gives much better precision and performance than the formula in \cite{Gander1989815}, since pseudoinverse computation is avoided and full size singular value decomposition is replaced with a much smaller eigenvalue decomposition. For the handle-based deformation, the relevant matrices were decomposed in a preprocessing step. We did not invest much effort to use the best possible decomposition and carefully tune the parameters. Specifically, we used LDL decomposition for the initial mesh approximation step, but a sparse QR decomposition for the global steps in the ARAP/ASAP deformation, due to numerical instabilities caused by LDL there.

\textbf{Limitations}. Our assumption is that the initial PM has planar faces. Otherwise, many of the calculations made are not well-defined. Of course, the planarity of faces can only be up to some numerical precision. We have found that the affine case is less sensitive to non-planar faces than the other cases. The mesh in Fig. \ref{fig:11} does not have planar faces, yet the eigenshapes computed for it in the affine space do not cause them to be "less" planar. On the other hand, the eigenshapes of the parallel case (not shown) quickly deteriorate the quality of the mesh. 

\textbf{Creating an initial PM}. The linear subspaces described here need an initial PM realizing a given topology to be generated from. The simplest way to generate such a PM is to take a non-polyhedral mesh with the given topology and project it to a plane. The original mesh can then be projected into a linear space generated from the flat mesh. The result of this, however, is usually unsatisfactory and we did not use it. Most of the PM's in this paper were created by experimenting with the TopMod 3.0 software \cite{topmod}, where we used the variety of subdivision schemes implemented there to create elaborate meshes from simple solids. If only the mesh topology is given, then a simple "spring-based" embedding, such as Tutte's \cite{citeulike:8044754}, should suffice. 

\textbf{Selecting the Right Space}. There are, literally, an uncountable number of linear subspaces available for a single PM. Even if we limit ourselves to the three cases mentioned above, the number of possibilities to assign them to faces is exponential and manually assigning them is tedious. We did not investigate methods to find the optimal linear subspace to work with, or even attempt to define what exactly optimal means. A simple definition could be: the subspace with the highest dimension. Experimentally we observed that in many cases the parallel space had the largest dimensional. However, this subspace does not generate much visual variation in the overall look of the PM, compared to the other spaces. This problem remains open for now, and we reserve it for future work. In practice, switching between the non-mixed cases provided sufficient variation.

Currently we use a number of heuristics while experimenting with our system. The affine space is easier to work with when there are many DOFs, as is the case for quad meshes with boundaries. In situations where the number of DOFs is too small, this is usually caused by faces with more than four edges or vertices of degree three. These can be automatically reassigned to the other two cases to achieve more freedom. On the other hand, when using the parallel or the third case, some faces may enjoy too much freedom and misbehave while deforming. These can be reassigned to the parallel case, since it better preserves the shape of a polygon.

A related problem is how to interpolate PM that are not related by a single linear space. We have shown that any two can be connected by a succession of three linear spaces, which is not very useful for interpolation. An interesting thing to try is to approximate paths in the manifold of PM's by linear segments using the linear subspaces. 

\textbf{Design pipeline}. Our experiments led us to the following pipeline for designing a PM. It is important to remember that we are mere computer scientists, and not artists. For flat meshes, the first step is to afford them some height. This is done by regular deformation followed by a planarization step, or by using the affine linear subspace and applying the handle-based deformation or using the eigenshape band-pass-filter technique. The reason for not using the parallel or vertical subspaces is that they cannot "unflatten" the PM. However, mixed spaces can also be used. Once we have a PM with some volume, the rest depends on the effect we aim to achieve. For large deformations we use the affine subspace when working on quad meshes with boundaries, and the other subspaces otherwise. To add variation or \emph{waviness} to the PM, we use the eigenshapes. The affine eigenshapes are useful when the overall look of the PM needs to be maintained but the shapes of individual faces need to be changed. Using the parallel eigenshapes is an efficient way of adding variation to meshes having uniformly-sized faces. We show a variety of results in Figs. \ref{fig:9}, \ref{fig:10}, \ref{fig:11} (see also the accompanying video).

The sparse and fundamental shapes, while helping to visualize the limitations of various subspaces, have not proven very useful for the design process. The reason is that, by definition, they can only make the PM less smooth, which usually means less visually pleasing. However, we believe they are valuable as a theoretical tool for studying PM's. One future research direction could be to use them to decide where to make small adjustments to the topology of the mesh in order to add more freedom to specific places.

\begin{figure}[th]
\centering
\includegraphics[width=1\linewidth]{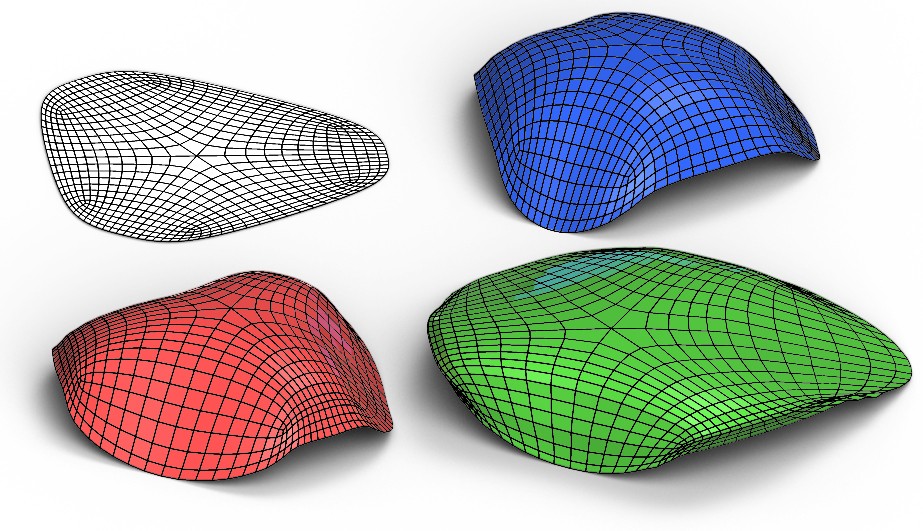}
\caption{Designing a PM from a planar graph. The graph was created by projecting a non-PM to the plane. It was then given height using the affine subspace, and then deformed using eigenshapes in the parallel and vertical spaces.}
\label{fig:10}
\end{figure}

\begin{figure}[th]
\centering
\includegraphics[width=1\linewidth]{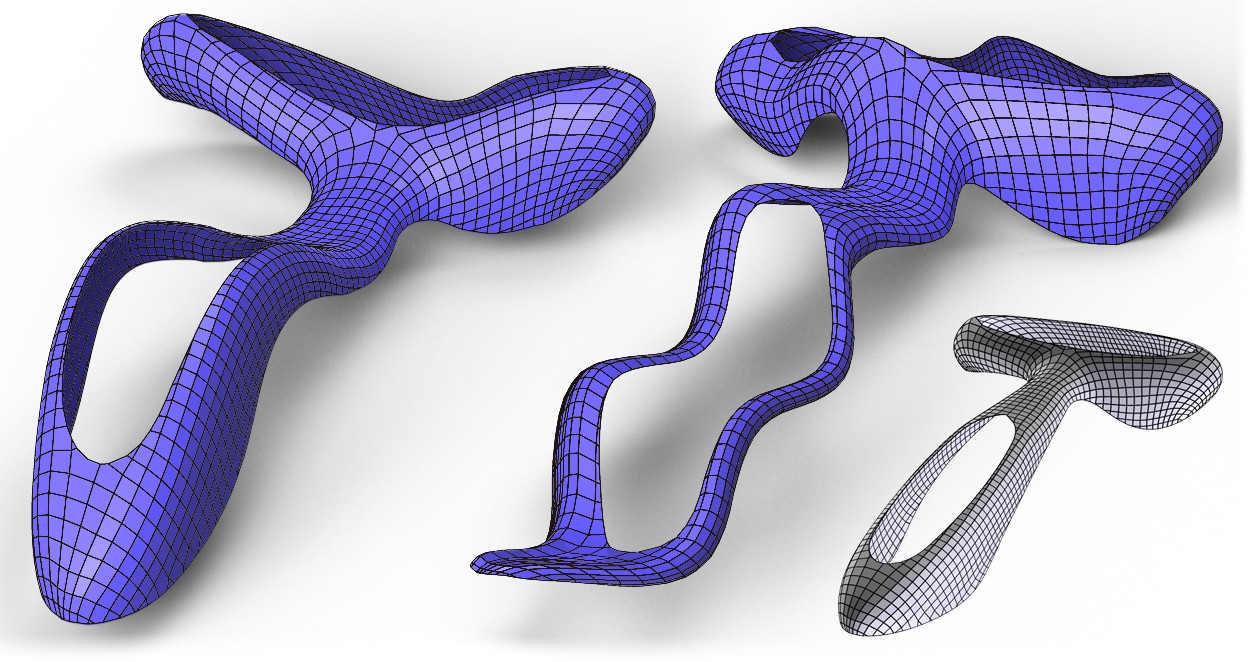}
\caption{The Yas model deformed using eigenshapes of different frequency in the affine subspace.}
\label{fig:11}
\end{figure}

\begin{figure}[th]
\centering
\includegraphics[width=1\linewidth]{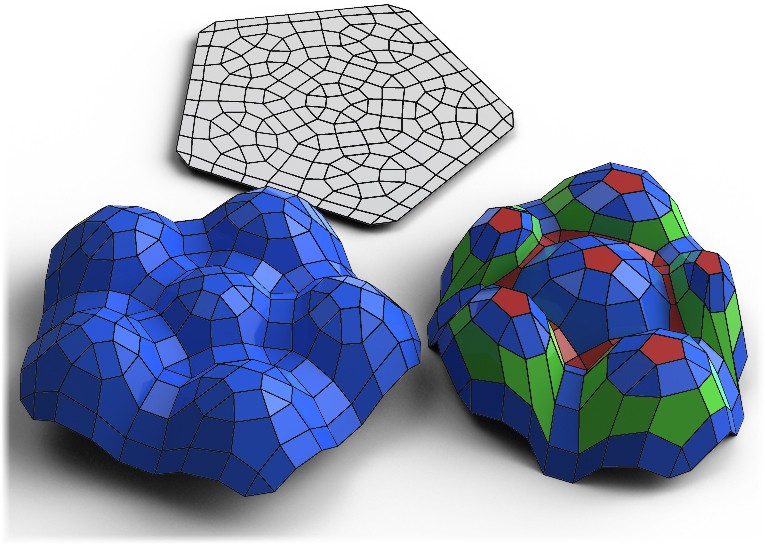}
\caption{Another example of designing a PM from a planar graph. The graph was created by subdividing a pentagon using several schemes until the desired result was achieved. It was given height in the affine subspace and then deformed in a mixed subspace.}
\label{fig:12}
\end{figure}

\begin{figure}[th]
\centering
\includegraphics[width=1\linewidth]{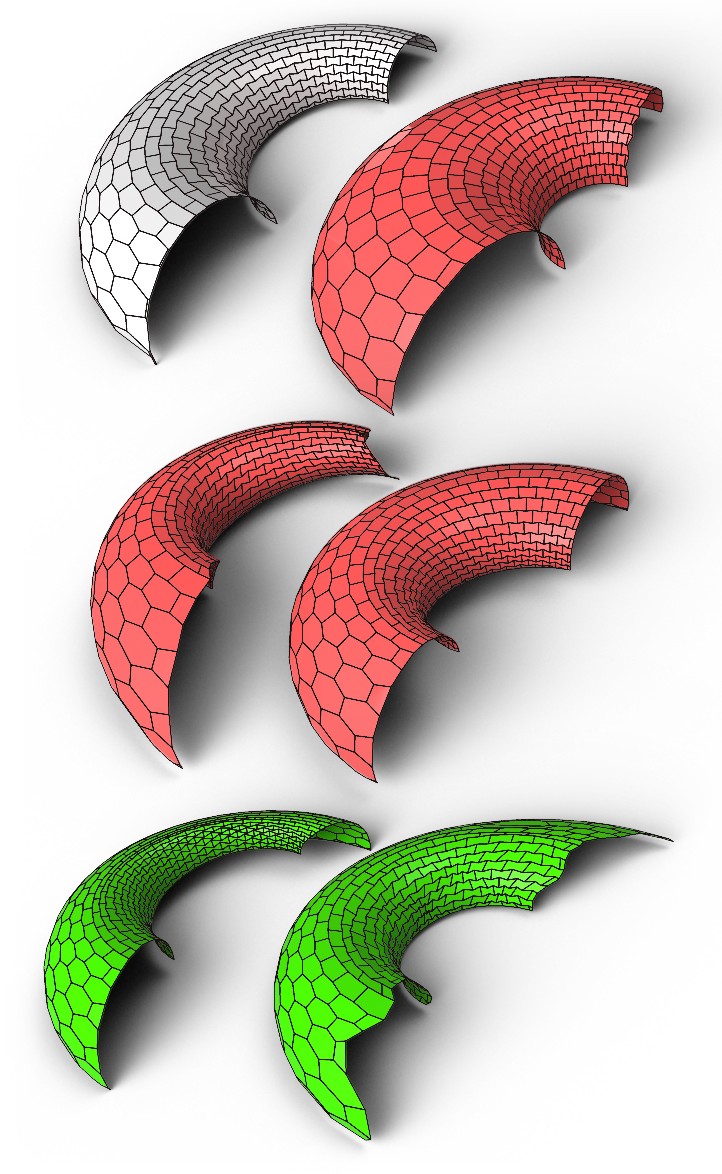}
\caption{Deformation of a planar hex mesh \cite{wang2008hexagonal} using eigenshapes in the parallel and vertical spaces. Note that the affine space had only 12 DOF, while the other cases had a few hundreds.}
\label{fig:13}
\end{figure}

\bibliographystyle{IEEEtranS}

\bibliography{linearpm}

\section*{Appendix 1}
We show a relation to differential equations in the case of the affine space. Assume a regular quad grid in the $xy$-plane and consider just the $z$ direction. Then $L_B$ can be written as a stencil:
$$L_B=\left[ {\begin{array}{*{20}{c}}
{1/4} &{-1/2}&{1/4} \\
{-1/2}&{1}   &{-1/2}\\
{1/4} &{-1/2}&{1/4}
\end{array}} \right]$$
For a 2D function $u$, this stencil is used to find the discretization of the mixed second derivative $u_xy$. Just as PM's are those that satisfy $L_B X=0$ locally, functions that satisfy $u_{xy}=0$ can be viewed as the continuous version of PM's. These are all of the form $u(x,y)=f(x)+g(y)$, where $f$ and $g$ are any functions. Hence we see that in the continuous case, as in the discrete case, the solution is determined by its values on one side of the boundary.

\section*{Appendix 2}
Recall that an orthogonal dual of an embedded planar graph is an embedding of the dual graph such that primal and dual edges are orthogonal. Also recall that a \emph{lifting} of a planar graph is a movement of its vertices in the $z$-direction such that the faces remain planar. Any PM in the affine space of a PM can be reached by consecutively lifting the source PM in the $x$,$y$ and $z$ directions. Hence, the dimension of the affine space is less than three times the dimension of the space of graph liftings. There is a known 1-1 correspondence between orthogonal duals and graph liftings up to translation\cite{richter1996realization} and we can take advantage of their simple geometric representation to get more insight on the space of liftings. Consider any 3-regular graph. Its dual faces are all triangles. Being an orthogonal dual, each dual triangle must have a fixed shape and orientation. In addition, all triangles must be translated and scaled together. Therefore, the entire space of orthogonal duals is determined by 3 values (translation and scale). Therefore, liftings of 3-regular graphs have only 3 DOFs and as a consequence, the affine space of 3-regular graphs has dimension at most 12.

\end{document}